\documentclass[journal]{IEEEtran}

\usepackage{cite}
\usepackage{amsmath,amssymb,bm,bbm}
\usepackage{cite}     
\usepackage{amsmath,color}   
\usepackage{xcolor}
\usepackage[mathscr]{euscript}
\usepackage[pdftex]{graphicx}
\usepackage{epstopdf}
\usepackage{mathtools}

\ifCLASSOPTIONcompsoc
  \usepackage[caption=false,font=normalsize,labelfont=sf,textfont=sf]{subfig}
\else
  \usepackage[caption=false,font=footnotesize]{subfig}
\fi

\usepackage{algorithm}
\usepackage{algpseudocode}

\newcommand{\norm}[1]{\left\Vert{#1}\right\Vert}
\newcommand{\CRB}{Cram\'{e}r Rao Bound }

\newcommand{\Ebb}{\mathbb{E}}

\usepackage{amsthm}
\usepackage{url}        
\begin{document}

\title{Multi-sensor Spatial Association using Joint Range-Doppler Features}
\author{Anant Gupta,~\IEEEmembership{Student Member,~IEEE}, Ahmet Dundar Sezer,~\IEEEmembership{Member,~IEEE}, and Upamanyu Madhow,~\IEEEmembership{Fellow,~IEEE}
\thanks{This work was supported in part by the Semiconductor Research Corporation (SRC) and DARPA under the JUMP program, the Center for Scientific Computing at UCSB, and the National Science Foundation under grants CNS-1518812, CNS-1518632 and CNS-1725797.}
\thanks{A. Gupta, A. D. Sezer, and U. Madhow are with the Department
of Electrical and Computer Engineering, University of California Santa Barbara, CA USA (e-mail: anantgupta@ucsb.edu, adsezer@ucsb.edu, madhow@ece.ucsb.edu)} 
}
\markboth{For Submission to IEEE Transactions on Signal Processing}{}
\maketitle

\begin{abstract}

We investigate the problem of localizing multiple targets using a single set of measurements from a network of radar sensors.  Such ``single snapshot imaging'' provides timely situational awareness, but
can utilize neither platform motion, as in synthetic aperture radar, nor track targets across time, as in Kalman filtering and its variants.  Associating measurements with targets becomes a fundamental bottleneck
in this setting. In this paper, we present a computationally efficient method to extract 2D position and velocity of multiple targets using a linear array of FMCW radar sensors by identifying and exploiting inherent geometric features to drastically reduce the complexity of spatial association. The proposed framework is robust to detection anomalies, and achieves order of magnitude lower complexity compared to conventional methods. While our approach is compatible with conventional FFT-based range-Doppler processing, we show that more sophisticated techniques for range-Doppler estimation lead to reduced data association
complexity as well as higher accuracy estimates of target positions and velocities.


	 
\end{abstract}

\section{Introduction}
\IEEEPARstart{R}ECENT advances in low-cost design and fabrication enable the potential application of high-accuracy millimeter wave (mmWave) radar sensors 
to a variety of commercial sectors, including automotive, drones and robotics \cite{ti1,ti2}.  The large available bandwidths enable high range resolution, while
the small wavelength enhances Doppler and microDoppler resolution. In this paper, we explore the utility of a network of such sensors in providing timely situational awareness for highly dynamic environments, by considering estimation of the kinematic state of the scene (i.e., the positions and velocities of targets) via a single set of measurements obtained by a network of sensors. We do not rely on tracking targets across time, or on platform motion to synthesize larger apertures.   

	
The specific problem we consider is that of localizing multiple targets in a 2D scene using a linear array of radar sensors. Figure \ref{fig:lin_array_fig} shows a scenario with two targets being observed with a linear array of four spatially separated sensors positioned along \textsf{x}-axis. Each sensor collects the relative range and Doppler observations for the targets in the scene. Since these observations are not ordered {\it a priori}, each range-Doppler measurement must first be associated with a target, and then the measurements associated with a given target from multiple sensors can be used to estimate its position and velocity.  Since the number of possible associations grows exponentially in the number of sensors, it is critical to develop efficient algorithms for spatial association. It is also important to build in robustness to missed detections, since
millimeter waves can be easily occluded by objects in the scene.


\begin{figure}[htbp]
	\begin{center}
	\includegraphics[width=\columnwidth]{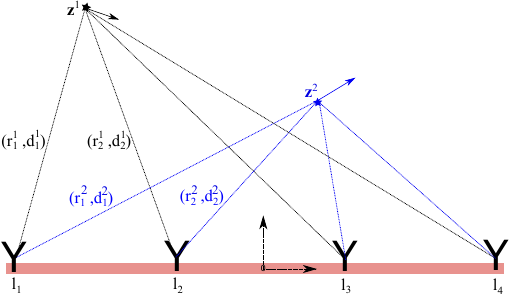}
	\caption{2D System model with a linear array of radar sensors placed on \textsf{x}-coordinates, $[l_1, l_2,l_3,l_4]$. The kinematic states $\bm{z}_1,\bm{z}_2$  of two targets are to be estimated using the unordered range and doppler observations from the sensors.}
	\label{fig:lin_array_fig}
	\end{center}
	\vspace{-0.6cm}%
\end{figure}


\subsection{Contributions}
Our goal is to develop robust and computationally efficient algorithms for single snapshot spatial data association. The main contributions of our study are as follows: \\
(1) We examine the geometric relations between instantaneous range, Doppler, and sensor locations, and show that features obtained via those geometric relations simplify the association problem. Specifically, we observe and exploit linear relationships between functions of the range-Doppler observations for a target across the linear array of sensors.\\
(2) We provide a low-complexity solution for the association problem by introducing a new graph-search based algorithm which prunes the set of feasible associations based on geometric relationships. In particular, our proposed algorithm considers a cost function based on the linear geometric relationships together with the triangle inequality constraint for the range observations at pairs of sensors and eliminates a significant number of possible associations. In addition, our approach accounts for detection anomalies such as missed detections and false alarms while reducing the complexity.\\
(3) We compare our proposed algorithm against conventional algorithms in the literature and evaluate performance in terms of localization accuracy, cardinality errors, robustness, and complexity. Also, we show that using an enhanced accuracy estimation algorithm (i.e., NOMP \cite{mamandipoor2016newtonized}) instead of conventional FFT-based approach improves localization accuracy and reduces association complexity as the number of targets and sensors increases. 

\subsection{Related Work} \label{rel_work}
The majority of prior work addresses {\it temporal} data association, which focuses on association of new measurement with existing target tracks. 
A number of techniques have been proposed in this regard, including Random Finite Set based sampling methods \cite{granstrom2017likelihood}, fuzzy clustering \cite{wang2018clustering}, and convex optimization \cite{williams2018multiple}. These methods rely on the temporal continuity of target state to assist in associating observations across multiple time frames. Most of these methods are designed for a single sensor case, and extensions to multiple sensor settings are not well-known. In this paper, our focus is on {\it spatial} association, where the data from multiple spatially separated sensors needs to be associated within the same time frame.
While the problem of spatial association studied in this paper has received relatively less attention, we provide a brief overview of the most widely used algorithms in the literature that can be extended for the spatial problem.

%
    
The association problem between a pair of sensors can be optimally solved using the well-known Hungarian algorithm \cite{kuhn1955hungarian}. However, a naive extension to multiple sensors by factorizing into pairwise (2D) associations over consecutive sensors does not work well in the presence of detection anomalies such as miss, false alarm, clutter, and close-target interactions \cite{wu2006tracking}. 
    
The multi-sensor association problem can be formulated as the Maximum A-Posteriori (MAP) estimation of most likely chain of observations across sensors. In order to solve this problem, a graphical model is defined, where a node represents sensor detection and an edge between nodes represents association hypothesis with a certain probability \cite{zhang2008global}. Association between sensors is obtained by solving the Minimum Cost Maximum Flow (MCF) problem over this graph. A variety of methods such as Linear Programming \cite{jiang2007linear}, Dynamic Programming \cite{pirsiavash2011globally, berclaz2011multiple}, and push-relabel maximum flow \cite{zhang2008global} have been proposed to efficiently solve the MCF problem. 
Although those methods solve the optimization in polynomial time, they require specialized mechanisms such as expansion of the observation set over successive iterations to resolve detection anomalies. Moreover, the complexity of the MCF problem grows quickly as $O(N^3\log N)$, where $N$ is the number of sensors \cite{zhang2008global}. 
In comparison with prior work, our approach reduces complexity by leveraging the high accuracy of sensor observations and their geometric properties. 
     
Probabilistic approaches such as the gated Nearest Neighbor (NN) \cite{bar1995multitarget} method sequentially associate observations across the sensors. At each sensor, each observation is associated with its closest match to the state predicted by the chain of observations from the past sensors. However, using only the single most likely observation to form association is vulnerable to clutter and anomalies in noisy scenarios. In addition, a single association error can cause significant contamination in final state estimate. This problem is well known in the literature on Simultaneous Localization and Mapping (SLAM), and various improvements such as Multiple Hypothesis Tracking \cite{reid1979algorithm}, K-best assignment \cite{murty1968letter}, and JPDAF \cite{bar1995multitarget} have been proposed.
In contrast, we propose an alternative search approach based on geometric fitting criteria which do not depend on such probabilistic models and avoid the contamination of state.

Bottom up approaches based on grid search over a set of candidate target states have been suggested in the literature \cite{folster2005data}. In \cite{venkateswaran2012localizing}, an approach based on enumerating all possible candidates followed by pruning and merging shows promising results. Randomized adaptive search procedures such as random consensus sampling (RANSAC) \cite{fischler1981random}, Interpretation Tree \cite{grimson1987localizing}, Joint Compatibility Branch and Bound \cite{neira2001data} have been shown to mitigate the impact of detection anomalies. These methods utilize a suitably defined metric to check the consistency of a set of associated observations, and employ branch and bound type search strategies to reduce the search complexity. Our graphical approach uses similar pruning techniques to perform the graph search,
but with the additional use of geometric constraints and a geometric fitting error metric for guiding the search.

\vspace*{5pt}
\noindent  \textbf{Outline:}
The rest of the paper is organized as follows. In Section \ref{prob_desc}, we introduce the association problem in the single snapshot localization setting. In Section \ref{graphical_association}, our graph association algorithm is presented. Then, the proposed algorithm is evaluated over different system parameters in Section \ref{sim_results} and Section \ref{conclusion} concludes the paper.
   
\vspace*{5pt}
\noindent   \textbf{Notation:}
$a, \bm{a}, A, \mathscr{A}$ represent scalar, vector, matrix and set respectively. We use $[.]$ to construct vector, matrix and $\{.\}$ to construct set. $\times, \cup, \cap$ denote the cartesian product, union and, intersection of two sets and $\varnothing$ denotes a \emph{NULL} value. $n(\mathscr{A})$ represents the number of non-empty elements in set $\mathscr{A}$. $\circ$ denotes element-wise multiplication between vectors. $A^T$ denotes transpose of matrix $A$ and $\wedge$ denotes logical ``and'' operator.

\newcommand{\NS}{N_S} 
\newcommand{\NT}{N_T} 

\section{Problem Description}
\label{prob_desc}

\subsection{System Model}
Consider a linear array of $\NS$ radar sensors in a two-dimensional (2D) scene with $\NT$ targets as in Figure~\ref{fig:lin_array_fig}. Without loss of generality, we assume that the sensor array is static and located along \textsf{x}-axis and centered at origin. The absolute kinematic state of the targets can be obtained by using the target location relative to this sensor array along with its own odometer information.

The kinematic state (i.e., instantaneous position and velocity information of all targets) of the scene is given by
\begin{align*}
    \mathcal{Z} = \{ \bm{z}^k \}_{k=1}^{\NT}
\end{align*}
where $\bm{z}^k = (x^k,y^k,v_x^k,v_y^k)$ is the kinematic state of target $k$ with an instantaneous velocity of $(v_x^k,v_y^k)$
at position $(x^k,y^k)$.

The range-Doppler of target $k$ observed at sensor $i$, can be expressed in terms of the desired kinematic state as follows,
    \begin{align}
    \label{range_doppler_transform}
         r_{i}^k=\sqrt{(x^k-l_i)^2+(y^k)^2},\quad d_i^k = \frac{(x^k-l_i)v_x^k+y^kv_y^k}{r_i^k} \,.
    \end{align}
 where $l_i$ is the \textsf{x}-coordinate of sensor $i$. We denote this non-linear mapping as $(r_i^k,d_i^k) = \mathcal{T}_i(\bm{z}^k)$.

\subsection{Single Snapshot Localization} 
In order to extract range and Doppler information of target $k$, each sensor $i$ uses the signal (i.e., $m_i^{obs}(t)$) reflected back from the scene in monostatic mode. In this study, we focus on localization of the scene using a single snapshot. For that reason, the kinematic state of the scene is assumed to be constant for a certain time interval and the scene localization is performed based on the range and Doppler information gathered during that time interval. Based on those, the Maximum Likelihood Estimator (MLE) for the scene including all $\NS$ sensors can be expressed as,
\begin{align}
\hat{\mathcal{Z}}_{ML} &= arg\max_{\mathcal{Z}} \prod_{i=1}^{Ns} \mathcal{L}\left(m^{obs}_i|\mathcal{T}_i(\mathcal{Z})\right) \label{general-llr}
\end{align}
where $m_i^{obs}$ corresponds to the observed signal in a single snapshot and $\mathcal{L}\left(m^{obs}_i|\mathcal{T}_i(\mathcal{Z})\right)$ is the conditional log likelihood of the observed signal for scene $\mathcal{Z}$. 

The optimization problem in \eqref{general-llr} is difficult in general since the number of targets (i.e., $\NT$) is not known and a brute force search for $\mathcal{Z}$ incurs exponential complexity in the number of targets; that is, $n(\mathcal{D}(\bm{z}))^{\NT}$ for a grid $\mathcal{D}(\bm{z})$. In addition, the observations contain a variety of anomalies such as clutter, missed detections, and false alarms, which further complicates the solution. 
    

In order to facilitate the solution of the problem in \eqref{general-llr}, the problem is divided into two stages as follows:

\subsubsection{Estimation}
The Range-Doppler pairs of $M_i \leq \NT$ non-occluded targets are estimated from received signal $m_i^{obs}$ at sensor $i$ using efficient algorithms proposed in the literature \cite{gupta2016super}.
The estimate at sensor $i$ for $k^{th}$ target can be modeled as follows,
\begin{subequations}
\label{assoc_model}
\begin{align}
    (r_i^{k}) = (r_i^k)^{true} + w_{i}^R + \tilde{b}^{k}_{i} \label{assoc_model_eq1}\,,\\
    (d_i^{{k}}) = (d_i^k)^{true} + w_{i}^D + \bar{b}^{k}_{i} \label{assoc_model_eq2}
\end{align}
\end{subequations}
where $w_{i}^R \sim \mathcal{N}(0,\sigma_{r_i}^2)$ and $w_{i}^D \sim \mathcal{N}(0,\sigma_{d_i}^2)$ denote independent Gaussian distributed noises with zero mean and $\tilde{b}^{k}_{i}$ and $\bar{b}^{k}_{i}$ denote the bias errors introduced due to proximity with any other $M_i-1$ targets in the scene. 
The noise variance depends on estimation accuracy at the given SNR which, in turn, depends on target radar cross section (RCS), path loss, and antenna directivity. For simplicity, we assume equal received signal power across all targets in the scene.
 
 We denote the set of estimated range-Doppler pairs at sensor $i$ by $\Theta_i = \{\bigcup_{k=1}^{M_i}\bm{\theta}_i^k \} $ where $\bm{\theta}_i^k = \left[(r_i^j),(d_i^j)\right]^T$. Index $k\in [1,M_i]$ in $\theta_i^k$ denotes the index of $k^{th}$ measurement with respect to the $M_i$ measurements for sensor $i$, whereas, index $j$ in $[(r_i^j ),(d_i^j )]^T$ denotes the global index of the $j^{th}$ target. The different superscripts are used to highlight the fact that the order of targets for which the range-Doppler measurements are obtained at the sensors is unknown. Indeed, our aim in this work is to find the correct ordering/association of the range-Doppler measurements.
 
\subsubsection{Association problem} \label{MDAP}
The estimation of kinematic state $\mathcal{Z}$ requires the association of those un-ordered range-Doppler pairs, $\Theta_i$, collected across all sensors. 
An association chain is defined as the ordered set of range-Doppler observations, $\mathscr{A}:\{ \{ \bm{\theta}_i \}_{i=1}^{\NS}| \bm{\theta}_i \in \tilde{\Theta}_i\}$ which is constructed from the \emph{NULL} augmented sets; that is, $\tilde{\Theta}_i = \Theta_i \cup \varnothing$. $\bm{\theta}_i=\varnothing$ corresponds to the \emph{NULL} state, which represents absence of observation at sensor $i$ (e.g., due to missed detection). Figure~\ref{fig:graph_assoc} shows a graphical representation of an association problem with three targets observed using $\NS=4$ sensors. Sensors $1$ and $2$ observe all targets $M_1=M_2=3$, sensor $3$ misses target $\bm{z}^2$ and sensor $4$ has a false alarm. The desired association chain for target $\bm{z}^1$ across four sensors is shown in the shaded region.

The spatial association problem can be formulated as the following maximum a posteriori (MAP) estimation problem,
    \begin{align}
       \bm{\mathscr{A}}^* =& \underset{\bm{\mathscr{A}} \subset \tilde{\Theta}_1 \times \cdots \times \tilde{\Theta}_{\NS}} {\text{argmax}} \log P(\bm{\mathscr{A}}) P(\Theta|\bm{\mathscr{A}}) \label{map-form} \\
       & \quad\text{such that} \quad \mathscr{A}^i \cap \mathscr{A}^j = \varnothing \quad \forall i \neq j, \quad n(\mathscr{A}^k)\geq 2 \nonumber
    \end{align}
where $\Theta = \bigcup_{i=1}^{\NS} \Theta_i$ denotes the set of all range-Doppler observations, $\bm{\mathscr{A}} = \{\mathscr{A}^1, \mathscr{A}^2, \cdots \}$ denotes a subset of association chains chosen from the set of all possible potential chains, $\tilde{\Theta}_1\times \tilde{\Theta}_2\times \cdots \tilde{\Theta}_{\NS}$.
The optimal solution $\bm{\mathscr{A}}^*$ consists of the set of chains which jointly maximizes overall log likelihood while the constraints ensure that no two chains share a common observation and each chain contains at least two observations.

When the targets are well-separated, the bias terms in \eqref{assoc_model_eq1} and \eqref{assoc_model_eq2} vanish and the likelihood for the individual targets becomes independent across multiple targets. In this case, the log likelihood in \eqref{map-form} simplifies to
    \begin{align*}
       \log P(\bm{\mathscr{A}}) P(\Theta|\bm{\mathscr{A}}) &= \sum_{\mathscr{A}\in \bm{\mathscr{A}}} \log P(\mathscr{A}) + \log P(\Theta|\mathscr{A})
    \end{align*}
where $P(\Theta|\mathscr{A})=\prod_{i=1}^{\NS} P(\Theta_i|\mathscr{A})$ is the probability of detecting the range-Doppler pairs which can be modeled by a Bernoulli distribution,
    \begin{align*}
    P(\Theta_i|\mathscr{A})= \begin{cases}
    \alpha & \text{, if target missed at sensor } i\,, \mathscr{A}_i=\varnothing\\
    1-\alpha & \text{, else }
    \end{cases}
\end{align*}     
where $\alpha$ denotes the probability of detection errors in \eqref{LLR} and is set to nominal value $\alpha=0.05$ \cite{zhang2008global}. This model accounts for the occurrence of both miss and false alarms across the sensors in the likelihood, which is given by 
\begin{align} \nonumber
 P(\Theta|\mathscr{A}) = \alpha^{\NS-n\left(\mathscr{A}\right)}(1-\alpha)^{n(\mathscr{A})}\,.
\end{align}
   
Also, $P(\mathscr{A})$ is the likelihood of chain modeled using the perceived range-Doppler pairs, $(\hat{r}_i, \hat{d}_i)=\mathcal{T}_i\left(\hat{\bm{z}}\right)$ for a target state $\hat{\bm{z}}_{k}$ predicted by the chain (see Section \ref{state-prediction}). By ignoring the constant terms which preserve the MAP solution, we define the normalized negative log likelihood as follows,
    \begin{align}
    \label{LLR}
        \mathcal{L}(\mathscr{A}) =& \sum_{\bm{\theta}_i\in\mathscr{A}}\left(\frac{(\hat{r}_i-r_i)^2}{\sigma_r^2} + \frac{(\hat{d}_i-d_i)^2}{\sigma_d^2}\right) \nonumber\\
        &+n(\mathscr{A})\log \frac{\alpha}{1-\alpha}
    \end{align}
where $\bm{\theta}_i=[r_i,d_i]^T$ is the observation from $i^{th}$ sensor in association chain $\mathscr{A}$ and $\sigma_r^2$ and $\sigma_d^2$ are the nominal variance terms for range and Doppler, respectively (see Appendix~\ref{crb_rd} for details). The first term in \eqref{LLR} denotes the squared error between the estimated and observed range-Doppler pairs in the chain while the second term penalizes the selection of smaller chains which prevents formation of duplicate chains for the same target. Hence, the association problem is reduced to the following constrained minimization problem:
    \begin{align}
    \bm{\mathscr{A}}^* =& \underset{ \bm{\mathscr{A}} \subset \tilde{\Theta}_1 \times \cdots \times \tilde{\Theta}_{\NS}} {\text{argmin}} \sum_{\mathscr{A} \in  \bm{\mathscr{A}}}\mathcal{L}(\mathscr{A}) \label{llr-form-joint} \\
       & \quad\text{such that} \quad \mathscr{A}^i \cap \mathscr{A}^j = \varnothing \,\, \forall i \neq j, \quad n(\mathscr{A})\geq 2 \nonumber
    \end{align}
 The joint minimization problem over all potential association chains in \eqref{llr-form-joint} is difficult in general. For that reason, we use an iterative approach where the most likely chains of observations are identified and removed from observation set $\Theta$ sequentially, 
 \begin{align}
    \underset{ {\mathscr{A}} \in \tilde{\Theta}_1 \times \cdots \times \tilde{\Theta}_{\NS}} {\text{argmin}} \mathcal{L}(\mathscr{A}) \quad\text{such that} \quad n(\mathscr{A})\geq 2 \,. \label{llr-form} 
    \end{align}
Without any prior knowledge of association between the nodes, the number of potential chains  $\Theta_1\times \Theta_2\times \cdots \Theta_{\NS}$ still grows exponentially. 
However, the formulation in \eqref{llr-form} enables the utilization of various network optimization methods to identify the most likely chain. Once the associated chains of range-Doppler observations are found across sensors, the kinematic state of the scene can be easily obtained by solving the inverse kinematic problem $[\hat{x},\hat{y},\hat{v}_x,\hat{v}_y]=\mathscr{T}^{-1}(\mathscr{A})$ using Gauss-Newton algorithm \cite{gupta2016super}.

%

\section{Graphical Association}\label{graphical_association} 
In order to solve the association problem in \eqref{llr-form}, we formulate the spatial association problem using graphical models and present our low-complexity graphical search method to obtain association chains efficiently via geometric relations.
    
\subsection{Graph Generation} \label{sec:graph_gen}
\begin{figure}[bp]
\vskip -1em
  \begin{center}
    \includegraphics[width=0.8\columnwidth]{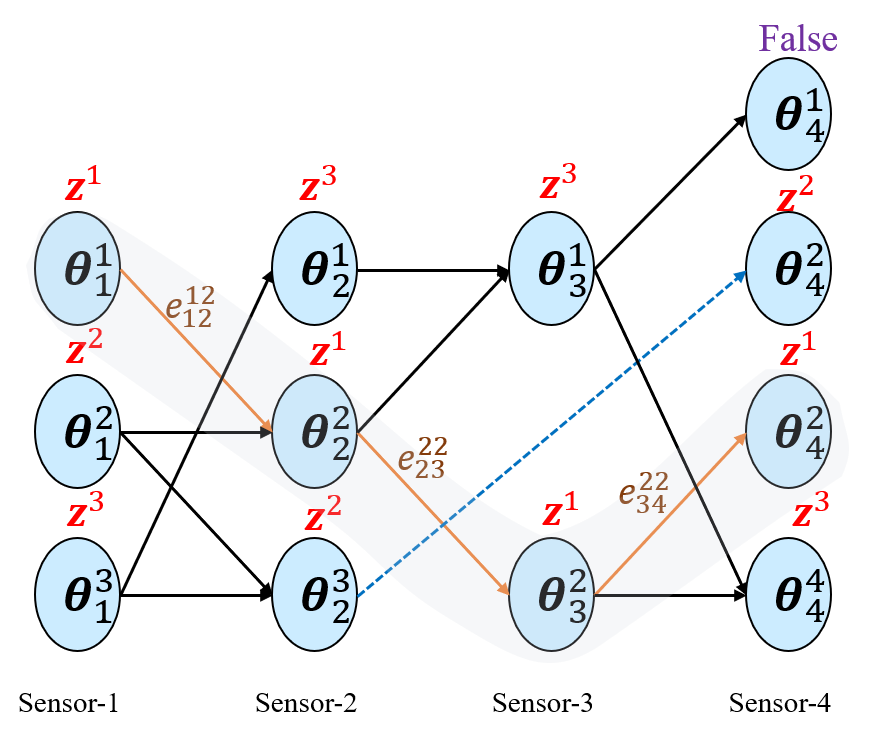}
    \caption{Target-based observation graph for a scene with 3 targets and 4 sensors. Sensors $1,2$ observe all 3 targets in different orders. Sensor 3 misses the observation of target state $\bm{z}^2$ while sensor 4 contains a false observation. Desired association chain, $\mathscr{A}$ is shown by the shaded set of nodes.}
    \label{fig:graph_assoc}
  \end{center}
\end{figure}
	
To begin with, we define a target-based graph to perform data association with following elements:
\begin{itemize}
    \item Node $\bm{\theta}_i^k$ represents the $k^\text{th}$ range-Doppler pair at sensor $i$. Nodes for a given sensor are arranged along a single column of the graph as shown in Figure \ref{fig:graph_assoc}.
    \item Edge $e_{ij}^{kl}=[\bm{\theta}_i^k, \bm{\theta}_j^l]$ denotes the linkage between pairs of observation across sensor $i$ and sensor $j$, which can correspond to a feasible target $\bm{z}_{ij}^{kl}$ referred to as ``\emph{candidate}'' location. 
    \item Chain $\mathscr{A}^j$ is represented by the sequence of two or more nodes spanning distinct sensors, which is associated to a single target, $\hat{\bm{z}}^{j}$.
\end{itemize}

\textbf{Geometric Constraint:} A significant portion of the edges can be easily discarded in the graph generation phase by using the following geometric constraint on target's range (for noiseless case),
\begin{align}
\label{prune1}
  C_{\textrm{G}}(e_{ij}):\quad (r_i - r_j < l_{ij})  \wedge (r_i + r_j > l_{ij})
\end{align}
where $l_{ij}=|l_i-l_j|$ represents the separation between sensor $i$ and sensor $j$.

Graph $\mathscr{G}=(\mathscr{V},\mathscr{E})$ is initialized with vertices for all range-Doppler pairs $\mathscr{V}=\left\{\Theta_i\right\}_{i=1}^{\NS}$ and edges $\mathscr{E}$ between any two consecutive nodes that satisfy condition $C_{\textrm{G}}(e_{i-1,i}^{k,l}), \forall k \in [1,M_{i-1}], \forall l \in [1,M_i]$ for all $i \in \{2,\ldots,\NS\}$ given in \eqref{prune1}.

\subsection{Spatial Association using Geometric Features}
In this subsection, we describe the solution of the association problem presented in \eqref{llr-form} using graph $\mathscr{G}$ by exploiting geometric relations between range, Doppler, and sensor geometry. For clarity of exposition, we focus on the association procedure of a single target $\bm{z}=[x,y,v_x,v_y]$ and, therefore, drop the superscript $k$ for the sake of simplicity.
    
\subsubsection{Geometric Relations}
  The range of target observed at $i^{th}$ sensor is given by 
  \begin{align}
  	\label{r2fit}
      r_i=\sqrt{(x-l_i)^2 + (y)^2}\,.
  \end{align}
  The Doppler component is the rate of change of range and it is given by,
  \begin{align}
      d_{i} = \dot{r}_{i} &= \frac{(x-l_i)v_x + yv_y}{r_{i}}\nonumber\\
      r_{i}d_{i} &= (x-l_i)(v_x) + yv_y\label{rdfit}\,.
  \end{align}
  For a linear array of sensors, the range and Doppler measurements for a target satisfy the following relations based on \eqref{r2fit} and \eqref{rdfit}:
  \begin{subequations}
  \label{rd_eq}
  \begin{align}
      r_{i}^2 &= r_j^2 - 2x(l_{i}-l_j)+(l_i^2-l_{j}^2) \label{r2quad}\\
      r_id_i &= r_jd_j - (v_x)(l_{i}-l_j)  \label{rdlin}
  \end{align}
  \end{subequations}
  where $r_i$ ($r_j$) and $d_i$ ($d_j$) are the range and Doppler estimated at the $i^{th}$ ($j^{th}$) sensor, respectively. $l_i$ ($l_j$) is the \textsf{x}-coordinate of $i^{th}$ ($j^{th}$) sensor. \eqref{r2fit} and \eqref{rdfit} indicate that for the noiseless setting, the range-Doppler products and range squared are linear with respect to target's velocity and position at \textsf{x}-coordinate, respectively. Therefore, the correct associations can be identified by fitting the observations to those geometric relations.
   
\subsubsection{State Prediction and Fitting Error} \label{state-prediction}
  The presence of noise in $(r_i,d_i)$ causes high error in these geometric relations due to the quadratic dependence. An estimate of target state parameters $\hat{x}, \hat{v}_x$ can be obtained by minimizing that error between observed and predicted range and Doppler values. Let $\bm{q}_1=\left[r_id_i | (r_i,d_i)\in \mathscr{A}\right]$ and $\bm{l}= \left[l_i| \bm{\theta}_i\in \mathscr{A}\right]$ denote the vector of range-Doppler products using observations in chain $\mathscr{A}$ and the vector of corresponding sensor \textsf{x}-coordinates, respectively. Predicted fit $\hat{\bm{q}}_1$ can be expressed using the geometric relation in \eqref{rdlin} as follows:
  \begin{align*}
  \hat{\bm{q}}_1 = -{v}_x\bm{l}+\kappa_1\mathbbm{1}=H\bm{s}_1
  \end{align*}
  where $H=[\bm{l}, \mathbbm{1}]$, $\bm{s}_1 = \left[-{v}_x \,\, \kappa_1\right]^T$, and $\kappa_1$ is a constant.
  Then, the least squares estimate for $\bm{\hat{s}}_1$ is obtained as
  \begin{align}
  \hat{\bm{s}}_1 &= arg\min_{\bm{s}_1} \qquad  \norm{\bm{q}_1 - H\bm{s}_1}^2 \label{max-fit-error}\\
  & =(H^TH)^{-1}H^T\bm{q}_1 \,. \nonumber
  \end{align}

Therefore, the least squares estimate is obtained as $\hat{v}_x=\bm{u}^T\bm{q}_1$ where $\bm{u}=-H(H^TH)^{-1}\bm{e}_1$ and $\bm{e}_1=[1,0]^T$.

Similarly, let $\bm{q}_2=\left[r_i^2 | (r_i,d_i)\in \mathscr{A}\right]$ denote the vector of range squared observations in chain $\mathscr{A}$, predicted fit $\hat{\bm{q}}_2$ can be expressed using the geometric relation in \eqref{r2quad} as follows:
\begin{align*}
\hat{\bm{q}}_2 -  \bm{l}\circ\bm{l} = -2{x}\bm{l}+\kappa_2\mathbbm{1}=H\bm{s}_2
\end{align*}
where $\bm{s}_2 = \left[-2x \,\, \kappa_2\right]^T$ and $\kappa_2$ is a constant. The least squares estimate of $\hat{x}$ is obtained as
\begin{align}
\hat{\bm{s}}_2 &= arg\min_{\bm{s}_2} \qquad \norm{\bm{q}_2-\bm{l}\circ\bm{l} -H\bm{s}_2}^2 \label{max-fit-error2}\\
&=(H^TH)^{-1}H^T(\bm{q}_2-\bm{l}\circ\bm{l}) \,. \nonumber
\end{align}
Hence, we obtain $\hat{x}=\bm{u}^T(\bm{q}_2-\bm{l}\circ\bm{l})/2$. 

The remaining state parameters (i.e., $\hat{y}$ and $\hat{v}_y$) are obtained using the geometric relations in \eqref{r2fit} and \eqref{rdfit} as
    \begin{align*}
        \hat{y}&=\sqrt{\frac{1}{n\left(\mathscr{A}\right)}\sum_{\bm{\theta}_i \in \mathscr{A}}\left(r_i^2-(\hat{x}-l_i)^2\right)}\,,\\             
        \hat{v}_y &=\frac{1}{n(\mathscr{A})}\sum_{\bm{\theta}_i \in \mathscr{A}}\frac{r_id_i-(\hat{x}-l_i)\hat{v}_x}{\hat{y}} \,.
    \end{align*}

The normalized geometric fitting error of a chain $\mathscr{A}$ can be computed using these estimates as follows:
    \begin{align} \label{err_fit}
      F(\mathscr{A}) & = \frac{\norm{\bm{q}_1-\hat{\bm{q}}_1}^2}{\eta_1} + \frac{\norm{\bm{q}_2-\hat{\bm{q}}_2}^2}{\eta_2} \\\label{err_fit2}
      &= \!\begin{multlined}[t]
        \frac{\norm{\left( I-H(H^TH)^{-1}H^T \right)\bm{q}_1}^2}{\eta_1} \\
        + \frac{\norm{\left(I-H^T(H^TH)^{-1}H\right)(\bm{q}_2-\bm{l}\circ\bm{l})}^2}{\eta_2}
      \end{multlined}
    \end{align}
where $\eta_1$ and $\eta_2$ are normalization constants that are set based on CRB (see Appendix~\ref{chi_squared_approx} for details) and \eqref{err_fit2} is obtained by substituting the predicted fits into \eqref{err_fit}. It is important to note that the error in \eqref{err_fit2} is additive over the observations in chain $\mathscr{A}$. Therefore, the extension of the chain cannot reduce the fitting error. In other words, $F(\mathscr{A})$ is monotonically non-decreasing over the length of chain $\mathscr{A}$. For that reason, the fitting error provides a simple measure of the geometric consistency of a chain, which can be used to traverse the graph and extract the chains efficiently.

\subsubsection{Geometric Association} \label{geometric-association}
We now present a graph search procedure which obtains the associated chains by minimizing geometric fitting error $F(\mathscr{A})$ in \eqref{err_fit2} and negative log likelihood $\mathcal{L}(\mathscr{A})$ in \eqref{LLR}. We apply the geometric relations by adding constraints on the desired chain, $\mathscr{A}$ to the optimization problem in \eqref{llr-form} as follows,
  \begin{subequations}\label{constraints}
  \begin{align}
   \underset{{\mathscr{A}} \in \tilde{\Theta}_1 \times \cdots \times \tilde{\Theta}_{\NS}} {\text{min}} &\mathcal{L}(\mathscr{A}) \nonumber\\
    \text{such that} \quad & n(\mathscr{A})\geq \gamma, \label{length_costraint}\\
     & F(\mathscr{A})< \tau^{n(\mathscr{A})}_{f} \label{cost_constraint}
  \end{align}
  \end{subequations}
The constraint in \eqref{length_costraint} restricts the number of missed observations to be less than $\NS-\gamma$ and the constraint in \eqref{cost_constraint} only allows chains with good geometric fit to be selected. In order to provide a solution for the optimization problem in \eqref{constraints}, we perform Depth First Search (DFS) over the graph generated in Section \ref{sec:graph_gen} to extract the chains, where those additional constraints help in reducing the search complexity. Our complete Spatial Association using Geometry Algorithm (SAGA) is outlined in Algorithm \ref{relax-assoc}. Here is a brief description:
  \begin{enumerate}
  \item We start the graph search by setting $\gamma=\NS$ so that only chains that include observations from all sensors are extracted. For that reason, we consider a graph having edges between consecutive sensors only. This helps to reduce the chains encountered during initial DFS procedure (see Appendix~\ref{consecutive-ambiguity} for details). 

  \item The DFS is guided by geometric fitting error $F(\mathscr{A})$. After each node is visited, the fitting error of candidate chain is calculated and the chain is ignored if it has a fitting error higher than predefined threshold $\tau_f^{\NS}$. Since the fitting error is non-decreasing over the length of the chain, most of the candidate chains are eliminated before reaching at the end of the graph, which reduces the complexity further. Details of DFS are shown in Appendix~\ref{dfs-algo}. At the termination of the DFS, the corresponding chain of nodes is added to solution $\bm{\mathscr{A}}^{\dagger}$ if it satisfies all the constraints in \eqref{constraints} and the negative log-likelihood of the association chain is below a predefined threshold (i.e., $\mathcal{L}(\mathscr{A})< \tau^{n(\mathscr{A})}_{l}$). The nodes belonging to the selected chains are removed from the graph together with their corresponding edges to keep subsequent chains disjoint. 

  \item  In order to deal with missed detection cases at sensors, the minimum chain length constraint (i.e.,$\gamma$) is relaxed in steps upto robustness level $\rho$. Due to that relaxation, the graph includes not only the edges between consecutive sensors but also the edges among the nodes that skip over $h$ consecutive sensors. Those edges are called \emph{Skip}-$h$ edges where $h = \NS - \gamma$. Then, the DFS procedure is repeated for different minimum chain length constraints. Consequently, in this procedure, \emph{NULL} states are taken into account and the generated chain does not include any observation from a sensor that misses the corresponding target by skipping over the observations of that sensor via \emph{Skip}-$h$ edges. In addition, the DFS procedure implicitly accounts for \emph{NULL} state in the beginning and end of a chain by starting searching from different nodes in consideration of minimum chain length constraint.

  \item The thresholds (i.e., $\tau_f$ and $\tau_l$) for the geometric fitting error and the likelihood depend on length of the chain $n(\mathscr{A})$ and their initial value is set based on CFAR criteria (see Appendix~\ref{chi_squared_approx} for details). Using a tight initial threshold ${\tau_f}$ for $F(\mathscr{A})$ restricts the number of branches to be explored at each node to a smaller set. This reduces the initial complexity of DFS while allowing only a subset of association chains $\bm{\mathscr{A}}^{\dagger} \subset \bm{\mathscr{A}}^*$ to be found. The thresholds are later relaxed by a factor of $\beta>1$ to allow the observations contaminated with noise to be selected. The relaxation is stopped when no further chains with length $n(\mathscr{A})\geq \NS-\rho$ exist in the graph.
  \end{enumerate}

  \begin{algorithm}
  \caption{Spatial Association using Geometric Assistance}\label{relax-assoc}
    \begin{algorithmic}[1]
    		\Statex{\textbf{Input:} Graph $\mathscr{G}$, Robustness level $\rho$}
        \State \textsc{Initialize} Chains $\bm{\mathscr{A}}^{\dagger}=\emptyset, \bm{\tau} =[{\tau_f, \tau_l}]_{init}$
        \Repeat
        		\State \textsc{Remove} all Skip edges
        		\For{$h=0$ to $\rho$}
        			\State Set minimum chain length: $\gamma=\NS-h$
        			\State \textsc{Add Skip-$h$ edges to graph $\mathscr{G}$}
        			\For{$v \in \mathscr{V}$ }
        				\State \textsc{DFS} from node $v$: $\mathscr{A} \gets \textit{GA-DFS}(v,\gamma,\bm{\tau}$)
					\If {Valid Chain, $\mathscr{A}$ is found}
						\State $\bm{\mathscr{A}}^{\dagger} \gets\mathscr{A}$ 
        					\State Remove chain from graph $\mathscr{V}=\mathscr{V}-\{\mathscr{A}\}$
        				\EndIf
        			\EndFor
         			 
        		\EndFor
			\State Relax thresholds: $\bm{\tau} \gets \beta{\tau}$	    	
		\Until{Chains with length $n(\mathscr{A})\geq \NS-\rho$ exists in $\mathscr{G}$}
    \Statex \textbf{Output:} Selected chains $\bm{\mathscr{A}}^{\dagger}$
    \end{algorithmic}
\end{algorithm}%


\textbf{Robustness:} 
During chain length relaxation, a \emph{Skip} edge is added between the observations across sensor $i$ and sensor $q$ if
\begin{enumerate}
  \item Observations $\bm{\theta}_i$ and $\bm{\theta}_q$ satisfy the geometric constraint  $C_{\textrm{G}}(e_{iq})$ in \eqref{prune1}, and,
  \item The target state predicted by $\bm{\theta}_i$ and $\bm{\theta}_q$ differs by a predefined threshold $\tau_z$ from the ones predicted by using all observations on the paths that connect $\bm{\theta}_i$ and $\bm{\theta}_q$. 
  \begin{align} \label{C2}
      C_{\textrm{S}}(e_{iq}):\quad \norm{\hat{\bm{z}}_{\mathscr{A}_p}-\hat{\bm{z}}_{iq}} > \tau_{z}\,, \forall \mathscr{A}_p:\left\{\bm{\theta}_i,\bm{\theta}_q\right\}\in\mathscr{A}_p
  \end{align}
  where $\mathscr{A}_p$ is in the form of $\mathscr{A}_p=\{\bm{\theta}_i,\bm{\theta}_j,\cdots,\bm{\theta}_q\}$ with $\bm{\theta}_i$ and $\bm{\theta}_q$ at the edges of the path, $\tau_z$ is set based on CRB (see Appendix~\ref{crb_pv_derivation}), $\hat{\bm{z}}_{iq}$ indicates the predicted target state based on $\bm{\theta}_i$ and $\bm{\theta}_q$, and $\hat{\bm{z}}_{\mathscr{A}_p}$ shows the predicted target state using the observations in $\mathscr{A}_p$. 
\end{enumerate} 
Enforcing the condition in \eqref{C2} avoids the formation of multiple chains corresponding to the same target and avoids unnecessary increase in the number of edges. The number of skip connections introduced in the graph is controlled by the robustness level; that is, $0\leq \rho \leq (\NS-2)$, which sets the maximum number of missed detections that can be tolerated across the sensor array. In this way, addition of such edges provides a flexible mechanism to provide robustness against missed detection in the sensors while keeping search space in control.
     
\textbf{Complexity:} 
The non-decreasing property of $F(\mathscr{A})$ is used to discard unlikely chains in the early stages of DFS. This allows for rapid extraction of associations without requiring search over all possible chains in the graph. The minimum track length threshold, $\gamma$, is reset to its maximum value after each relaxation. Therefore, the skip edges in the graph can be removed at the end of the inner loop to reduce search complexity further. Therefore, our approach exploits the geometric structure of observations across multiple sensors to reduce search complexity. 

\subsection{Spatial Association using Edge-based State Likelihoods} Before evaluating the performance of our main algorithm, we describe an iterative search method, which relies on the fact that an approximate kinematic state estimate can be derived by using two connected observations in a graph. In other words, a state estimate can be obtained for each edge in a graph, which is a part of the association chain $\mathscr{A}$. Therefore, the search space for the association problem in \eqref{llr-form} can be reduced to the set of edges.  

The likelihood of a candidate $\bm{z}_e$ corresponding to an edge $e\in \mathscr{E}$ can be computed as, 
  \begin{align}\label{phantom-llr}
  \mathcal{L}(\bm{z}_e) = \sum_{i=1}^{\NS}\left[ \min_{\bm{\theta}\in\Theta_i} \left(\frac{({r}_i'-r_i)^2}{\sigma_r^2} + \frac{({d}_i'-d_i)^2}{\sigma_d^2}\right)\right]
  \end{align}
where $[r_i',d_i']=\mathscr{T}_i(\bm{z}_e)$ is the perceived range and Doppler at sensor $i$ for target state $\bm{z}_e$.
Then, the most likely candidate can be selected by evaluating \eqref{phantom-llr} over all edges and choosing the one that achieves the minimum negative log likelihood; that is,  $\bm{z}^* = \bm{z}_{e^{*}}$ for $e^{*} = arg\min_{e\in\mathscr{E}} \mathcal{L}(\bm{z}_e)$. Then, the observations associated with $\bm{z}^*$ can be identified via the following neighborhood constraint:
  \begin{align*}
    \mathscr{N}(\bm{z}^*) = \bigcup_{i=1}^{\NS}\{ (r_i,d_i) | (r_i-r_i^*) \leq \delta_r \wedge (d_i-d_i^*)\leq \delta_d \}
  \end{align*}
where $[r_i^*,d_i^*]=\mathscr{T}_i(\bm{z}^*)$ are the perceived range-Doppler at sensor $i$ and $\delta_r$ and $\delta_d$ are the range and Doppler resolution parameters defined in Appendix~\ref{crb_rd}. The algorithm carrying out this Spatial Association using Edge-based State Likelihoods (\emph{SAESL}) procedure is presented in Algorithm \ref{spatial-bruteforce}.
  \begin{algorithm}
  \caption{SAESL Algorithm}
  \label{spatial-bruteforce}
  \begin{algorithmic}[1]
      \State \textsc{Initialize Graph with observations $\Theta$}: $\mathscr{G}=(\mathscr{V},\mathscr{E})$
      \State \textsc{Augment Graph} with skip edges 
      \For{$h=0$ to $\rho$} 
    	\State \textsc{Add Skip-$h$ edges to graph $\mathscr{G}$}
      \EndFor
      \State \textsc{Initialize} $\mathscr{Z}=\emptyset$
      \While{$\mathscr{E}\neq\emptyset$}
          \State \textsc{Find most likely Candidate}, $\mathscr{Z}\gets\bm{z}^*$ 
          \Statex \hspace{0.5cm}from edge $\bm{z}^*=arg\min_{e\in\mathscr{E}} \mathcal{L}(\bm{z}_e)$
          \State \textsc{Remove all vertices explained by} $\bm{z}^*$, 
          \Statex \hspace{0.5cm}{$\mathscr{V}\gets\mathscr{V}-\mathscr{N}(\bm{z}^*)$}
          \State \textsc{Update edges $\mathscr{E}$} 
      \EndWhile
      \State \textsc{Return} Selected candidates $\mathscr{Z}$
      \end{algorithmic}
  \end{algorithm}%
        
Since all edges in the graph are checked while selecting the candidates, this approach exhibits higher complexity than our proposed algorithm. Moreover, evaluation of state likelihood $\mathcal{L}(\bm{z}_e)$ in \eqref{phantom-llr} is more expensive than evaluation of chain likelihood $\mathcal{L}({\mathscr{A}})$ in \eqref{LLR} as it involves a minimization over all other observations. In Section~\ref{sim_results}, we use this algorithm as a benchmark against our proposed algorithm.

\section{Simulation Results} \label{sim_results}

In this section, we evaluate the performance of the proposed spatial association algorithm, \emph{SAGA} against the \emph{SAESL} algorithm through various performance metrics.  We consider a linear array of $\NS$ FMCW radar sensors each of which collects range and Doppler observations from the scene. The FMCW radar system parameters are set based on the ones that are used in typical low cost automotive systems at mm-Wave frequencies \cite{wagner2013wide}: bandwidth $B=0.5$ GHz, carrier frequency $f_c=60$ GHz, chirp duration $78 \mu s$, $N_{ch}=64$ chirps, and sampling rate $f_s=0.82$ MHz.  This provides range and Doppler resolutions of $\delta r = 0.3$ and $\delta d = 0.5 m/s$, respectively, and maximum range and Doppler of $19.2m$ and $\pm 16 m/s$, respectively, suitable for short range situational awareness applications. 
In the simulations, a single snapshot of the scene is considered with multiple targets having equal received SNR at all sensors. The kinematic states of targets are randomly selected based on uniform distributions $x\sim\mathcal{U}(-8m, 8m), y\sim\mathcal{U}(2m, 12m), v_x \sim \mathcal{U}(-10m/s, 10m/s),v_y \sim \mathcal{U}(-10m/s, 10m/s)$. 

It is important to note that when range and Doppler separation between two targets gets small, the estimation algorithm either provides a merged estimate or results in detection anomalies such as miss and false alarm. In order to differentiate the scenes with such estimation errors due to range-Doppler proximity, we consider two different scenarios with two different scenes. The \emph{well-separated} scene is generated by enforcing a minimum separation between the range and Doppler of the targets at all sensors.
The \emph{adverse} scene does not have such constraints and contains additional missed detection anomalies by randomly removing measurements from the sensors with probability $P_{miss}$. 
Unless stated otherwise, the nominal values of system parameters are presented in Table~\ref{tab:simParameters}. 

	\begin{table}[!hbt]
		\begin{center}
		\caption{Simulation Parameters}
		\label{tab:simParameters}
		\begin{tabular}{|c|c|}
			\hline
			{Number of targets} & $\NT=20$ \\
			\hline
			{Number of radar sensors} & $\NS=6$ \\
			\hline
			{ SNR} & $-10$ dB\\
			\hline
			{ Sensor Array Width} & $L_W=4$ m\\
			\hline
			Simulated misses & $P_{miss}=0.05$\\
			\hline
			{Robustness Level} & $\rho=4$\\
			\hline
			{Max error threshold} & $\bar{d}=0.16$ m\\
			\hline
		\end{tabular}
		\end{center}
		\vspace{-0.6cm}
	\end{table}

\subsection{Localization Accuracy}
In this subsection, we analyze the localization accuracy of kinematic state estimates obtained using associated sensor observations. This depends on the accuracy of underlying range-Doppler estimates. The position and velocity estimation errors for state estimates $\hat{\bm{\mathcal{Z}}}$ are computed as follows:
    \begin{align*}
    D_p(\hat{\bm{\mathcal{Z}}}) = \frac{1}{n(\hat{\bm{\mathcal{Z}}})}\sum_{\hat{\bm{{z}}}\in \hat{\bm{\mathcal{Z}}}} \min_{\bm{z}\in \bm{\mathcal{Z}}^\text{true}} d_p(\bm{z}, \hat{\bm{{z}}})^2 \\
        D_v(\hat{\bm{\mathcal{Z}}}) = \frac{1}{n(\hat{\bm{\mathcal{Z}}})}\sum_{\hat{\bm{{z}}}\in \hat{\bm{\mathcal{Z}}}} \min_{\bm{z}\in \bm{\mathcal{Z}}^\text{true}} d_v(\bm{z}, \hat{\bm{{z}}})^2
    \end{align*}
where $d_p(\bm{z}, \bm{{z}}')=\sqrt{(x-x')^2 + (y-y')^2}$ and $d_v(\bm{z}, \bm{{z}}')=\sqrt{(v_x-v_x')^2 + (v_y-v_y')^2}$ are the errors in position and velocity, respectively. The CRBs for Range-Doppler and Position-Velocity estimates are evaluated in Appendix~\ref{crb_rd} and Appendix~\ref{crb_pv_derivation}, respectively. Figure \ref{fig:rd_snr} shows the Root Mean Square Error (RMSE) in range-Doppler estimated at sensor level for different number of targets in a \emph{well-separated} case. We observe that range-Doppler RMSE at individual sensors achieves CRB at a SNR$= -15\,$dB threshold. The RMSE for position-velocity estimates obtained from sensor observations also achieve their CRB at the same SNR threshold. This shows that association using \emph{SAGA} does not introduce any additional errors to the localization process when SNR is above this threshold. However, the RMSE increases sharply below the SNR threshold due to the difficulty in associating noisy range-Doppler pairs. Therefore, we use nominal $SNR=-10$ dB in our simulations to perform further analysis. 
    \begin{figure}[t]
		\begin{center}
		\includegraphics[width=\columnwidth]{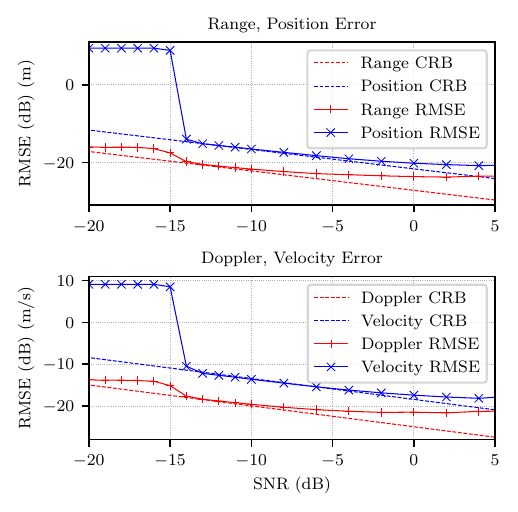}
		\vspace{-0.7cm}
		\caption{Range-Dopper estimation accuracy and Position-Velocity estimation accuracy versus SNR. The position-velocity RMSE converges to the CRB bound as SNR increases and the SNR at which this convergence occurs is called as SNR threshold. The SNR threshold provides an indicator for the localization performance when multiple targets, $\NT>1$ are present.}
		\label{fig:rd_snr}
		\end{center}
		\vspace{-0.6cm}
	\end{figure}
	
\textbf{Cardinality Error and OSPA}:  For multiple targets, the number of \emph{valid} targets identified by the system is also an important performance metric.  An estimated target $\hat{\bm{z}}$ is classified to be \emph{valid} only if it lies within a region ``close'' to the true targets, $\min_{\bm{z}\in \bm{\mathcal{Z}}^\text{true}}\norm{\hat{\bm{z}}-\bm{z}}<\bar{d}$ where $\bar{d}$ sets the maximum error threshold. The cardinality error is defined as the difference between actual number of targets and the number of estimated target; that is, $\NT-N_e=|\bm{\mathcal{Z}}^\text{true}| - |\bm{\hat{\mathcal{Z}}}|$. That error is caused due to the detection anomalies in the estimation algorithm at sensor level as well as during the association stage. In such cases, the localization accuracy by itself does not capture the true performance of the system. Therefore, we use the optimal subpattern assignment (OSPA) metric \cite{ristic2011metric}, which combines the localization and cardinality error into a single performance metric and is given by
	    \begin{align}\nonumber
	    \texttt{OSPA}(\hat{\bm{\mathcal{Z}}})=\sqrt{ \frac{1}{n(\hat{\bm{\mathcal{Z}}})}\left(\sum_{i=1}^m \min\left({d_c(\hat{\bm{{z}}_i}),\bar{d}}\right)^2 + |N_e-\NT|\bar{d}^2 \right) }
	    \end{align}
where $m$ is the number of \emph{valid} targets, $N_e-\NT$ is the cardinality error and, $d_c(\hat{\bm{{z}}_i}) $ is the localization error computed relative to \emph{closest} true target given as
\begin{align*}
d_c(\hat{\bm{{z}}_i}) = \min_{\bm{z}\in \bm{\mathcal{Z}}^\text{true}} d_p(\bm{z}, \hat{\bm{{z}}_i})^2 + d_v(\bm{z}, \hat{\bm{{z}}_i})^2\,.
\end{align*}
    
Figure~\ref{fig:OSPA_Nob} shows the OSPA error along with the localization and cardinality errors with increasing scene density in the \emph{well-separated} case. Both localization error and cardinality error start to increase as the scene gets denser. The \emph{SAGA} and \emph{SAESL} schemes have comparable localization error for each target that
is validated, but SAGA underestimates the number of targets (i.e., $N_e<N_T$) while SAESL overestimates it (i.e., $N_e>N_T$).
Since the localization error is computed only over the reduced set of \emph{valid} targets, we compute the OSPA metric which effectively combines both quantities.
The OSPA metric is dominated by localization error when the scene is sparse and cardinality errors when the scene is dense. We observe that \emph{SAGA} has slightly worse overall performance compared to \emph{SAESL} as the number of targets increases at low $SNR=-15$dB. This is because the geometric fit that SAGA relies on is impaired at lower SNR.
As we increase SNR to $-10$dB, the geometric fit and hence performance of SAGA improves,
reducing the performance difference with SAESL. Moreover, \emph{SAGA} obtains the association with significantly lower complexity than \emph{SAESL}, as discussed in the next section. 
    
\begin{figure}[htbp]
	\begin{center}
	\label{OSPA_Nob}
	\includegraphics[width=\columnwidth]{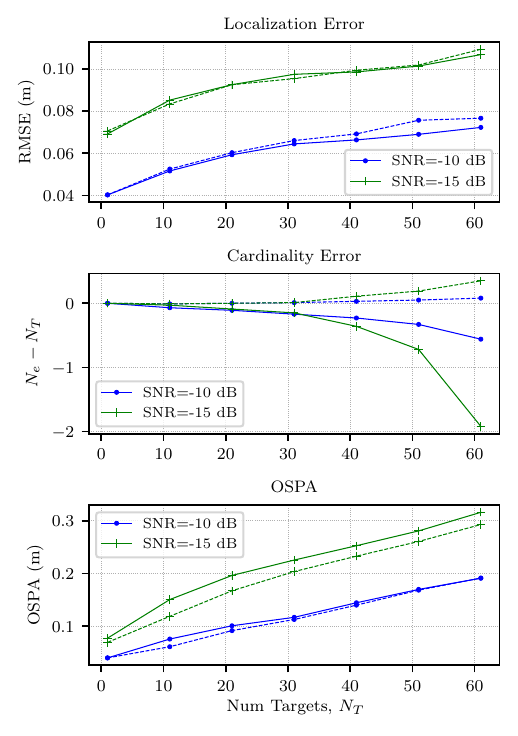}
	\vspace{-0.7cm}
	\caption{Overall localization accuracy versus number of targets at $SNR=-15,-10$ dB. The solid and dotted lines represent the performances of \emph{SAGA} and \emph{SAESL} association algorithms, respectively. }
	\label{fig:OSPA_Nob}
	\end{center}
	\vspace{-0.6cm}
\end{figure}
	
\subsection{Complexity Reduction} \label{complexity-reduction}
In this section, we analyze the computational savings achieved by the proposed \emph{SAGA} algorithm and provide comparison against traditional approaches. Figure~\ref{fig:graph_nodes} shows that the geometric pruning criteria in \eqref{prune1} reduce the number of edges in the graph by an order of magnitude as $N_T$ increases. Also, the worst case complexity of \emph{SAGA} can be expressed in terms of the number of association chains visited over the graph. We observe that the number of chains visited is lower than the number of total pruned edges in the graph and lies close to the lower bound $N_S N_T$. This shows the effectiveness of geometric features in solving the association problem with significantly lower complexity than the worst-case $N_T^{N_S}$ complexity.

In order to effectively compare the performance, we now consider adverse scenes in which the sensor observations contain detection anomalies. 
When the miss probability is low, \emph{SAGA} rapidly extracts all chains. As the number of missed detections increases, the robust scheme automatically increases the number of iterations by allowing relaxation of constraints in DFS graph search. In contrast, \emph{SAESL} always requires a large number of iterations. 
    
\emph{SAGA} provides robustness to missed detections by selectively adding skip edges to the graph. This mechanism reduces the OSPA error in adverse scenarios at the expense of increased computational complexity. The level of robustness can be tuned using a parameter $\rho$ which is set based on the adversity of the scene.  Figure~\ref{fig:graph_nodes} also shows the estimation performance for different robustness levels with increasing scene adversity (i.e., increasing miss detections). OSPA error reduces with higher robustness levels. However, low robustness level (e.g., $\rho=1$) is sufficient to obtain good performance at typical miss detection probability $P_{miss}<0.05$. Similarly, a higher robustness level helps to reduce the cardinality errors when the scene contains higher number of targets. The highest robustness level is $\rho=4$, which corresponds to the minimum chain length constraint in \eqref{length_costraint} with $n(\mathscr{A})\geq 2$.

\begin{figure}[tbp]
	\begin{center}
	\includegraphics[width=\columnwidth]{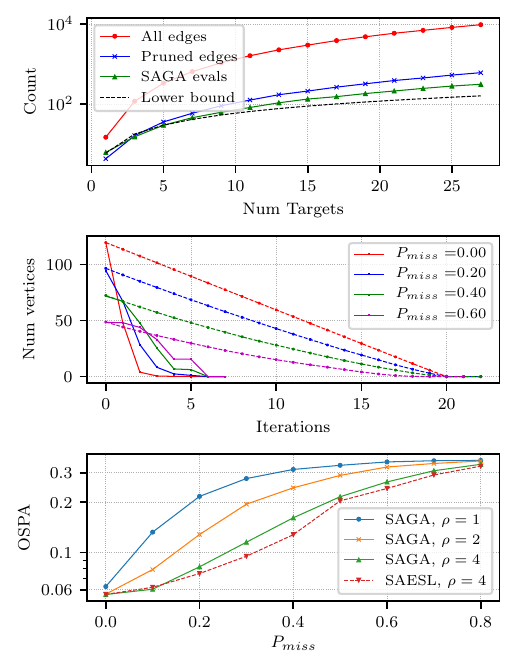}
	\vspace{-0.7cm}
	\caption{(Top) Graph edge count and number of association chains visited by SAGA algorithm. (Middle) Graph size at end of each iteration of association algorithm for different $P_{miss}$. \emph{SAGA} is denoted by solid line while \emph{SAESL} is denoted by dotted line.  (Bottom) OSPA versus $P_{miss}$ with different robustness levels $\rho$}
	\label{fig:graph_nodes}
	\end{center}
	\vspace{-0.6cm}
\end{figure}
	
\textbf{Runtime Comparison:} We now compare the computational complexity of our approach against \emph{SAESL}. Computing the number of operations that occur during association is difficult, since the number of chains visited depends on a variety of factors such as the fitting error thresholds and minimum chain length. However, given the same sensor estimates for the simulated scenes, we compare the relative complexities of \emph{SAGA} against other methods in Figure~\ref{fig:Cmplx_Nob-Rob} in terms of total number of Floating Point operations (FLOPS) and the total runtime. We observe that \emph{SAGA} exhibits an order of magnitude lower complexity compared to the \emph{SAESL} algorithm.  Moreover, this improvement increases as the number of targets increases, which highlights the advantage of our approach. In addition, as we increase the robustness level (e.g., from $\rho=0$ to $\rho=4$),  
the relative increase in complexity for the proposed \emph{SAGA} algorithm is far less than for the \emph{SAESL} approach.

\begin{figure}[htbp]
	\begin{center}
	\includegraphics[width=\columnwidth]{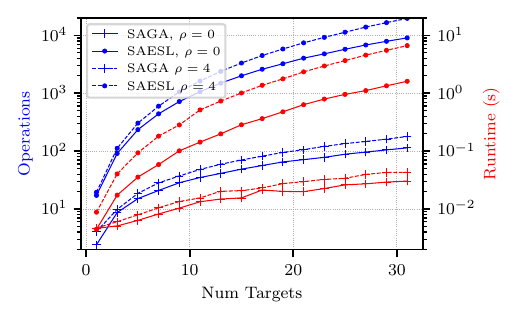}
	\vspace{-0.7cm}
	\caption{Association complexity versus the number of targets averaged over 100 trials using nominal parameters with robustness levels $\rho=0$ and $\rho=4$. Total number of FLOPS is denoted by blue line while the runtime is in red.}
	\label{fig:Cmplx_Nob-Rob}
	\end{center}
	\vspace{-0.6cm}
\end{figure}
	
We also compare the complexity against traditional methods such as gated Nearest neighbor filter (\emph{NN}) and Minimum cost flow (\emph{MCF}). The NN association scheme \cite{bar1995multitarget} builds the association chain by starting with a local kinematic state estimate from a pair of sensor observations and sequentially adding the nearest measurement from other sensors to update this state. The MCF association scheme \cite{zhang2008global} identifies the most likely set of chain by solving the minimum cost maximum flows over the graph. The cost of each edge is set based on its relative likelihood similar to our \emph{SAESL} method. We use an optimized implementation \cite{ortools} of \emph{MCF} for comparison purposes.
     
In order to compare the complexity of those algorithms, we count the number of times that the primary objective function (i.e., the likelihood cost in \eqref{llr-form}) is computed during the graph search procedure. Figure~\ref{fig:comparison_SOTA_complex} compares the complexity across algorithms as a function of the number of targets. The proposed \emph{SAGA} algorithm requires the lowest number of likelihood evaluations--significantly lower than for the naive \emph{SAESL} iterative search method. This shows our algorithm can effectively predict the correct chain using the geometric fitting criteria. The \emph{MCF} and \emph{NN} algorithms have similar complexities, lying between those of \emph{SAGA} and \emph{SAESL}.

	\begin{figure}[!t]
		\begin{center}
		\includegraphics[width=\columnwidth]{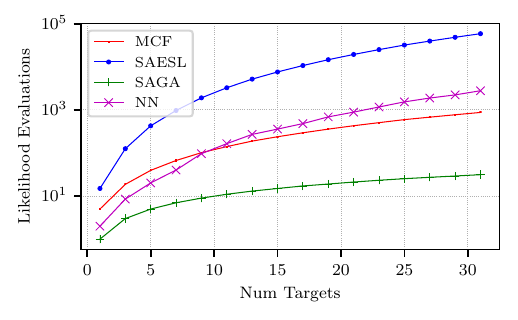}
		\vspace{-0.7cm}
		\caption{Number of evaluations of Likelihood\, $\mathcal{L}(\mathscr{A})$ with increasing number of targets. \emph{SAGA} has the lowest complexity across all scene densities while the complexities \emph{NN} and \emph{MCF} lie between those of \emph{SAGA} and \emph{SAESL}.}
		\label{fig:comparison_SOTA_complex}
		\end{center}
		\vspace{-0.6cm}
	\end{figure}

Figure~\ref{fig:comparison_SOTA_runtime} shows the overall runtime of the algorithms as the scene density increases. We observe that \emph{SAGA} is faster than the other methods by an order of magnitude. Since, the FLOPS count is not available for these other methods, we only compare the overall runtime, which follows a similar trend and provides a reasonable estimate of algorithmic complexity.
\begin{figure}[htbp]
		\begin{center}
		\includegraphics[width=\columnwidth]{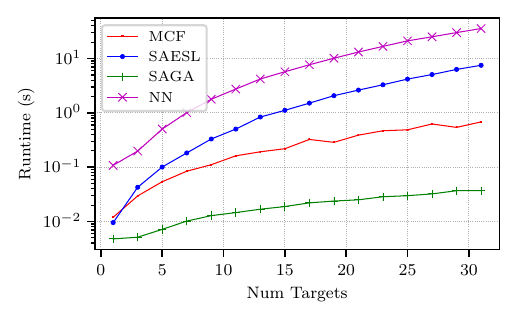}
		\vspace{-0.7cm}
		\caption{Runtime comparison with traditional algorithms.}
		\label{fig:comparison_SOTA_runtime}
		\end{center}
		\vspace{-0.6cm}%
\end{figure}
	

\subsection{Benefit of Super-Resolution} \label{dis:super}
Our algorithm extracts the geometric relationships between range-Doppler measurements based on the sensor array geometry and builds association chains by adding likely observations at new sensors to existing chains. In this section, we investigate the role of enhanced accuracy of range and Doppler estimates obtained using the NOMP \cite{mamandipoor2016newtonized} super-resolution algorithm in spatial association compared to coarse estimates obtained using DFT. Figure~\ref{fig:est_comparison_Nob} compares the localization and cardinality errors. We see that the localization accuracy using NOMP estimates achieves the CRB when the number of targets is moderate, whereas DFT has higher RMSE as expected. However, the RMSE of NOMP deviates away from CRB as the number of targets increases and approaches the accuracy of DFT-based estimates for dense scenes. 

It is important to note that our association algorithm works even with the coarse DFT-based estimates. However, NOMP provides an accuracy boost at the {\it input} of the association algorithm, which enables identification of more targets and results in lower cardinality errors compared to DFT.

The increased accuracy of NOMP estimates also results in smaller association time relative to DFT, due to the reduction in geometric fitting errors.
This reduction in association time comes, of course, at the expense of additional computation during range-Doppler estimation.	
Figure~\ref{fig:comparison_Nsens} compares the runtime of the estimation and association stages with different number of sensors for $N_T=20$ targets. We observe that the association time with NOMP estimates is $10$ times lower than the one with DFT estimates, while the estimation overhead is about $2-3$ times higher. 
Figure~\ref{fig:comparison_Nsens} shows that the computational complexity of association starts dominating that of estimation
as the number of targets and sensors increases. Thus, the overall complexity reduction due to NOMP-based estimation, relative to DFT-based estimation, becomes more pronounced
with a denser scene and a larger number of sensors.
\begin{figure}[htbp]
	\begin{center}
	\includegraphics[width=\columnwidth]{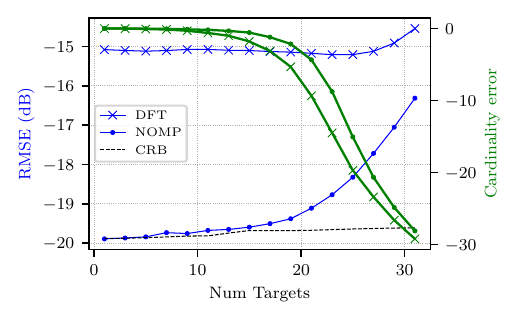}
	\vspace{-0.7cm}
	\caption{Estimation accuracy (thin) and cardinality error (thick) versus number of targets at SNR = -15 dB.}
	\label{fig:est_comparison_Nob}
	\end{center}
	\vspace{-0.6cm}
\end{figure}
\begin{figure}[htb]
	\begin{center}
	\includegraphics[width=\columnwidth]{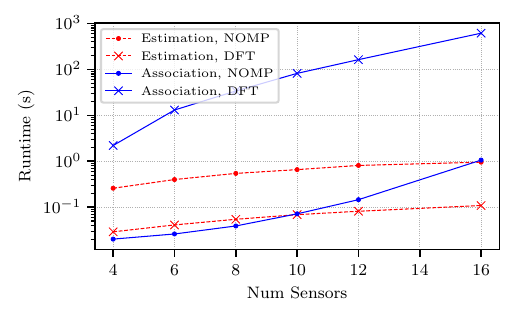}
	\vspace{-0.7cm}
	\caption{Runtime comparison of association (solid) \& estimation (dotted) stages versus number of sensors.}
	\label{fig:comparison_Nsens}
	\end{center}
	\vspace{-0.6cm}
\end{figure}
	

\subsection{Array Geometry} \label{placement}
In this subsection, we analyze the localization performance of linear sensor arrays from the perspective of data association.  We consider the adverse scene with $P_{miss}=0.2$ to emphasize our findings. The array width and the number of sensors affect both localization accuracy and association complexity.

Increasing the array width generates more spatial diversity in range-Doppler measurements across sensors. This helps to reduce the OSPA error for a given number of sensors. On the other hand, larger distance among the sensors weakens the pruning criteria for the graph edges used in \eqref{prune1}, resulting in a denser graph with a higher number of potential associations between sensors. Therefore, the overall localization performance improves with wider arrays at the expense of slightly more association complexity. The available sensor width is an important design constraint in practical applications (e.g., the length of the side profile of a vehicle). We therefore analyze the effect of the number of sensors, keeping the array width fixed to
 $L_W=4$ m.
\begin{figure}[b]
	\begin{center}
	\includegraphics[width=\columnwidth]{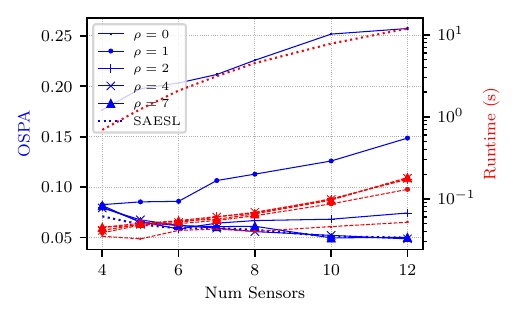}
	\vspace{-0.7cm}
	\caption{Association versus number of sensors for SAESL (thick dotted) and SAGA (thin solid) with different robustness levels.}
	\label{fig:design_Nsens_ospa}
	\end{center}
	\vspace{-0.6cm}
\end{figure}
	       
We find that increasing the number of sensors improves association performance as well as association complexity. Figure \ref{fig:design_Nsens_ospa} shows OSPA versus number of sensors for \emph{SAESL} and \emph{SAGA}. While the OSPA for \emph{SAESL} association decreases monotonically with the number of sensors, we observe that the OSPA for \emph{SAGA} with robustness level $\rho$ achieves minimum OSPA with $\NS=\rho+3$ sensors, and increases for $\NS>\rho+3$. This is due to missed observations preventing the formation of chains with minimum length constraint $\NS-\rho$. For an array with $\NS$ sensors and a robustness level of $\rho$, the expected number of missed targets can be expressed as
  \begin{align*}
  \Ebb[\mathrm{miss}] = \sum_{k=1}^{\min{(\NS-2,\rho+1)}} { \genfrac(){0pt}{}{\NS}{k}} P_{miss}^k (1-P_{miss})^{\NS-k}\,.
  \end{align*}
Figure \ref{fig:design_Nsens_card} shows that the number of missed targets observed using our approach closely matches this expected value for various values of $\rho$ and $\NS$. 

\begin{figure}[t]
\begin{center}
	\includegraphics[width=\columnwidth]{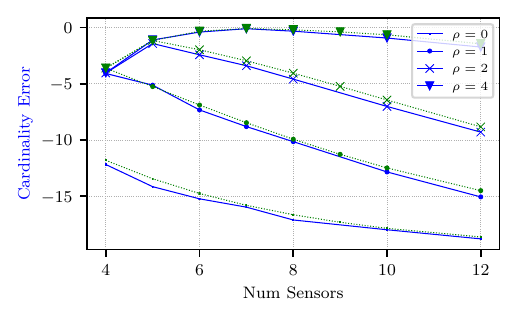}
	\vspace{-0.7cm}
	\caption{Comparison of simulated (solid) and theoretical (dotted) cardinality error.}
	\label{fig:design_Nsens_card}
\end{center}
\vspace{-0.6cm}
\end{figure} 

Thus, as we increase the number of sensors, while we improve localization accuracy, we must increase the robustness level used in the \emph{SAGA} algorithm
(setting it to $\rho = N_s - 3$) in order to avoid increase in cardinality errors. While this does result in increased computational complexity, it is still significantly lower than that of 
the \emph{SAESL} algorithm. We leave as an open issue the design of more sophisticated methods for selection of a subset of sensors during the association stage to reduce the complexity further.



\section{Conclusion} \label{conclusion}

We have shown that simple constraints relating range-Doppler observations to sensor geometry can be exploited to significantly reduce the complexity of spatial association. Our system-level simulations demonstrate that the proposed framework for spatial association based on these geometric constraints is robust to noisy observations and detection anomalies, and that it scales well with the number of sensors and targets. Our approach is compatible with standard
FFT-based range-Doppler processing, but enhanced accuracy estimation at each sensor (i.e., super-resolution of range and Doppler) significantly improves both localization accuracy and association complexity. The geometric constraints used to simplify the association problem rely on a linear placement of the sensor array, which is reasonable, for example,
when the sensors are placed along the side, front or back of a vehicle.  An interesting open question is whether such geometric concepts can be extended to simplify association for more general sensor array configurations. Additional
important topics for future investigation include extending these ideas to more complex target models (e.g., for extended targets, and targets causing both specular and diffuse reflection), and combining them with complementary strategies utilizing platform and/or target motion across multiple snapshots.


	
\appendices
\newtheorem{theorem}{Theorem}
\newtheorem{corollary}{Corollary}
\newtheorem{lemma}{Lemma}

\section{CRB for Range and Doppler} \label{crb_rd}
The \CRB provides an estimation theoretic lower bound on the sample covariance of range-Doppler estimates; that is, $Cov(\bm{\theta}_i) \geq I(\bm{\theta}_i)^{-1}$ where $I(\bm{\theta}_i)$ is Fisher Information Matrix (FIM) given by,
\begin{align}
I(\bm{\theta}_i)&=\Ebb\left[\left({\nabla_{\bm{\theta}_i} \mathcal{L}(m_i^{obs}|\bm{\theta}_i)}\right)\left({\nabla_{\bm{\theta}_i} \mathcal{L}(m_i^{obs}|\bm{\theta}_i)}\right)^H\right]\nonumber
\end{align}
where $\mathcal{L}(m_i^{obs}|\bm{\theta}_i)$ is the log likelihood of the observed signal for a given target range-Doppler $\bm{\theta}_i$. 
For an FMCW radar, this expression simplifies to \cite{gupta2016super},
\begin{align}
I(\bm{\theta}_i)&=  \kappa \left(\frac{A^2}{\sigma^2}\right) \begin{bmatrix}
1/{\delta_r^2} & 0\\
0 & 1/\delta_d^2
\end{bmatrix} \label{FIM-exp2}
\end{align}
where $\kappa$ is a constant, and $\delta_{r}$ and $\delta_{d}$ are the Rayleigh range and Doppler resolutions, respectively. $\frac{A^2}{\sigma^2}$ in \eqref{FIM-exp2} is the SNR of the received signal, $m_i^{obs}(t)$ with $m_i^{obs}(t)=Ae^{j\phi_{FMCW}(t)}+w(t)$ at sensor $i$, where $w(t)\sim \mathcal{N}(0,\sigma^2)$.
We set the nominal variance of range-Doppler estimates based on the value of CRB at $SNR=-20$ dB; that is, $\sigma_{r_i} = \frac{\delta_r\sigma}{\kappa A}$ and $\sigma_{d_i} = \frac{\delta_d\sigma}{\kappa A}$. 
\section{CRB for Position and Velocity} \label{crb_pv_derivation}
Using the range-Doppler model in Section \ref{assoc_model}, we evaluate the single target CRB for kinematic parameters $\bar{\bm{z}}$ using the log likelihood of range-Doppler observations $\mathscr{A}=\{\bm{\theta}_i\}_{i=1}^{\NS}$ given kinematic state $\bar{\bm{z}}$, which is
  \begin{align*}
  \mathcal{L}\left(\{\bm{\theta}_i|\bar{\bm{z}}\}_{i=1}^{\NS}\right) =& \sum_{i=1}^{\NS}\left( \frac{(\bar{r}_i-r_i)^2}{\sigma_{r_i}^2} + \frac{(\bar{d}_i-d_i)^2}{\sigma_{d_i}^2}\right) 
  \end{align*}
where $\bm{\theta}_i=(r_i, d_i)$ is the observed range-Doppler pair for sensor $i$, $(\bar{r}_i, \bar{d}_i)=\mathcal{T}_i(\bar{\bm{z}})$ is true range-Doppler pair for given target state $\bar{\bm{z}}$ and $\sigma_{r_i}^2$ and $\sigma_{d_i}^2$ are, respectively, the range and Doppler CRBs obtained in \eqref{FIM-exp2}. The FIM for $\bar{\bm{z}}$ can be evaluated as
  \begin{align*}
  I(\bar{\bm{z}})&=\Ebb\left[\nabla_{\bm{z}} \mathcal{L}\left(\{\bm{\theta}_i|\bar{\bm{z}}\}_{i=1}^{\NS}\right) \right] \,.
  \end{align*}
The CRB obtained from inverse FIM is used to find position and velocity CRB as follows,
  \begin{align*}
  CRB_p &= I(\bar{\bm{z}})^{-1}_{(1,1)} + I(\bar{\bm{z}})^{-1}_{(2,2)}\,,\\
  CRB_v &= I(\bar{\bm{z}})^{-1}_{(3,3)} + I(\bar{\bm{z}})^{-1}_{(4,4)}\,.
  \end{align*}
The CRB of velocity is a function of both range and Doppler variances whereas the CRB of position only depends on the variance of range. We use the nominal range and Doppler CRB values to set the minimum separation distance threshold, $\tau_z= 10 \sqrt{CRB_p+CRB_v}$ between targets. This threshold is also used to check similarity between chains in the association algorithm.
\section{Association Constraint Relaxation}
\label{chi_squared_approx}
The choice of initial stopping thresholds $\tau_f^n$ and $\tau_{l}^n$ and scaling factor $\beta$ for subsequent relaxations in \emph{SAGA} algorithm governs the total complexity of association algorithm. In order to initialize the association algorithm, we set tight thresholds for $\mathscr{L(A)}$ and $\mathscr{F(A)}$. Assuming the range-Doppler observations have small error (i.e., $w^R_{i}\ll r_i,\,w^D_{i} \ll d_i$ in \eqref{assoc_model}), the expected negative log likelihood in \eqref{LLR} can be approximated as
		\begin{align*}
		\mathcal{L}(\mathscr{A}) \approx & \sum_{\bm{\theta}_i \in \mathscr{A}}\left(\frac{(w^R_{i})^2}{\sigma_r^2} + \frac{(w^D_i)^2}{\sigma_d^2}\right) \,.
		\end{align*}
Since $w^R_i\sim \mathcal{N}(0,\sigma_r^2)$ and $w^D_i\sim \mathcal{N}(0,\sigma_d^2)$ are standard Normal distributed random variables, $\mathcal{L}(\mathscr{A}^k)$ has chi-squared distribution, $\chi^2_{2n(\mathscr{A})}$ with $2n(\mathscr{A})$ degrees of freedom. Then, the expected fitting error in \eqref{err_fit} can be approximated as
		\begin{align}
		F(\mathscr{A}) &= \sum_{\bm{\theta}_i \in \mathscr{A}} \frac{\left((\hat{r}_i\hat{d}_i)-(r_id_i)\right)^2}{\eta_1}
									+ \frac{\left((\hat{r}_i)^2-(r_i)^2\right)^2}{\eta_2}  \nonumber \\
									&\approx \sum_{\bm{\theta}_i \in \mathscr{A}} \frac{\left(r_iw^D_i + d_iw^R_i\right)^2}{\eta_1}
									+ \frac{\left(2r_iw^R_i\right)^2}{\eta_2} \label{fit_cost_approx}
		\end{align}
where $(\hat{r}_i$ and $\hat{d}_i)$ denote the perceived range-Doppler pair at sensor $i$ for predicted state $\bm{z}$ and $\bm{\theta}_i=({r}_i, {d}_i)$ denotes the observed range-Doppler pair at sensor $i$. The normalization factors $\eta_1, \eta_2$ are set to the variance of numerator terms which is,
        \begin{align*}
        (\eta_1)_i &= \text{Var}[r_iw^D_i + d_iw^R_i] \approx \sigma_{r_i}^2d_i^2 + r_i^2\sigma_{d_i}^2 + \sigma_{r_i}^2\sigma_{d_i}^2\\
        (\eta_2)_i &= \text{Var}[2r_iw^R_i] \approx 4r_i^2\sigma_{r_i}^2\,.
    \end{align*}
Using those values to normalize \eqref{fit_cost_approx} results in $\mathcal{F}(\mathscr{A}^k) \sim \chi^2_{2n(\mathscr{A})}$ being chi-squared distributed with $2n(\mathscr{A})$ degrees of freedom. Hence, the thresholds for the association algorithm are determined as follows,
\begin{align*}
\tau_f^{n(\mathscr{A})} : \mathrm{Pr}(F(\mathscr{A})>\tau_f^{n(\mathscr{A})})=P_\mathrm{FA}\\
\tau_l^{n(\mathscr{A})} : \mathrm{Pr}(\mathcal{L}(\mathscr{A})>\tau_l^{n(\mathscr{A})})=P_\mathrm{FA}
\end{align*}     
where $P_\mathrm{FA}$ is the nominal false alarm rate set to $P_\mathrm{FA} = 0.01$.

Note that while the normalization factors $\eta_1,\eta_2$ depend on $r_i, d_i$, we set this based on the maximum range, Doppler values to get a conservative initial value. This does not cause a problem since the sucessive relaxation procedure loosens that threshold so that chains with high fitting error can be extracted.

The relaxation factor, $\beta$ should be set appropriately. Choosing a high value causes faster convergence but might lead to false chains being identified. On the other hand, a low value delays the extraction of \emph{loose} chains. In the simulations, we find that $\beta=2$ performs well.
		

%

\section{Depth First Search}   \label{dfs-algo} 
A depth first search algorithm is outlined in Algorithm~\ref{ga-dfs}. At each node, the DFS procedure traverses through all branches which have geometric fitting error below the maximum error threshold $\tau_f^{\NS}$. On reaching the end of the graph, we select the chain if it satisfies the likelihood, fitting error, and minimum chain length constraints. In addition, we check for possible chain termination at each node after going through all its branches. This step implicitly accounts for the NULL state at the end of a chain.
     
\begin{algorithm}
  \caption{Geometry Assisted Depth First Search}
  \label{ga-dfs}
  \begin{algorithmic}[1]

  \Procedure{GA-DFS}{$v, \mathscr{A}, \gamma, \bm{\tau}$}
  	\State {Get list of children of $v$ that \emph{geometrically fit}, \Statex \hspace{1cm}$B(v) =\left\{v_j:\mathcal{F}([\mathscr{A}, v_j])<\tau_f^{\NS}\right\}$}
  	\If {$B(v) \neq \emptyset$}
  		\State Sort $B(v)$ using geometric fitting error, $F([\mathscr{A}, v_j])$
  		\For{ $v_j  \in B(v)$}
  			\State \textsc{Branch} out a new chain $\mathscr{A}^j:\mathscr{A}\gets v_j$
  			\State $\mathscr{A}^o\gets$\Call{GA-DFS}{$v_j, \mathscr{A}^j, \gamma, \bm{\tau}$}
  			\State Exit loop if valid chain $\mathscr{A}^o$ is found.
  		\EndFor
  	\EndIf
  	\State \textsc{Check if chain can be terminated at $v$}
  	\If {$n(\mathscr{A})\geq\gamma, \mathcal{L}(\mathscr{A})<\tau^{n(\mathscr{A})}_{l}, F(\mathscr{A})< \tau^{n(\mathscr{A})}_{f}$}
  		\State \textsc{Select} $\mathscr{A}^o\gets \mathscr{A}$,
  	\EndIf
  	\State \textbf{Output: }$\mathscr{A}^o$
  \EndProcedure
  \end{algorithmic}
\end{algorithm}%
\vspace{-0.4cm}
\section{Minimum Ambiguity Association} \label{consecutive-ambiguity}
\begin{lemma}\label{consec-amb-lemma}
In the ideal detection scenario (i.e., no miss or false alarms), the number of candidate locations generated between a pair of sensors is minimum for consecutive sensors.
\end{lemma}
\begin{proof}
Recall that candidate locations are generated when range perceived at a pair of sensors satisfy conditions in \eqref{prune1}. For a candidate, $\bm{z}_{ij}^{pq}$ generated by incorrectly associated observations, $\bm{\theta}_i^{p}, \bm{\theta}_j^{q}$, across consecutive sensors $i,j$, the following relations hold,
\begin{align}
\label{C1ghost}
r_i^p-r_j^{q} < l_{ij} ,\quad r_i^p+r_j^{q} > l_{ij} \,.
\end{align}
Now consider $\bar{q}^{th}$ observation at sensor $k$ adjacent to sensor $j$ which corresponds to same target as $\bm{\theta}_j^{q}$, the following hold,
\begin{align*}
r_j^{{q}}-r_k^{\bar{q}} &< l_{jk} \tag{using \eqref{prune1}} \\
r_j^{{q}}+l_{jk} &> r_k^{\bar{q}} \tag{$l_{jk}\geq0$, Triangle inequality}
\end{align*}
Using these along with \eqref{C1ghost} we obtain,
\begin{align*}
r_i^p-r_k^{\bar{q}} < l_{ik} ,\quad r_i^p+r_k^{\bar{q}} > l_{ik} 
\end{align*}
Hence any candidate produced between consecutive sensors $i,j$ also generates a candidate between sensors $i,k$ by skipping over intermediate sensor $j$. Hence,
\begin{align*}
\sum_{p=1}^{n(\Theta_i)} \sum_{q=1}^{n(\Theta_{i+1})} n(\bm{z}_{i,i+1}^{pq}) \leq \sum_{p=1}^{n(\Theta_i)} \sum_{q=1}^{n(\Theta_k)} n(\bm{z}_{i,k}^{pq})
\end{align*}
Therefore, the number of candidates generated between a pair of sensors is minimum for consecutive sensors.
\end{proof}
Association complexity is due to the presence of unwanted candidate targets which need to be discarded based on their likelihood. When a target is observed at all sensors, it is sufficient to associate observations along consecutive sensors. 
Lemma \ref{consec-amb-lemma} states that the association of observations along consecutive sensors generates the lowest number of phantoms during graph search. Hence, the number of potential ambiguities is minimized when the graph search procedure is conducted across consecutive sensors first. 


\bibliographystyle{IEEEtran}
\bibliography{articles}

%

\begin{IEEEbiography}[{\includegraphics[width=1in,height=1.25in,clip,keepaspectratio]{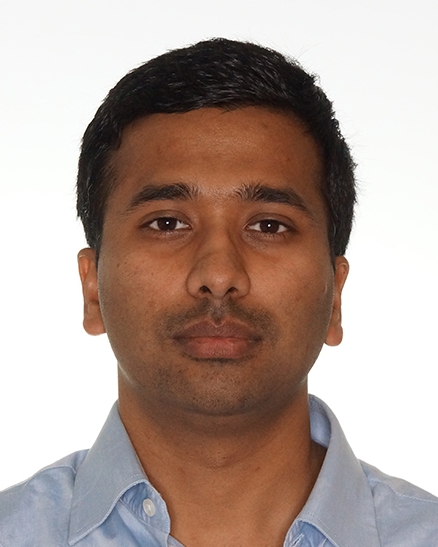}}]{Anant Gupta} received the B.Tech. degree in electronics and electrical communication engineering and the M.Tech. degree in telecommunication systems engineering from IIT Kharagpur in 2013. He received the M.S. and Ph.D. degree in electrical and computer engineering from the University of California at Santa Barbara (UCSB) in 2016 and 2020 respectively. He is currently working as Senior Engineer in the wireless R\&D team at Qualcomm, San Diego. His research interests include wireless sensing, signal processing and machine learning.
\end{IEEEbiography}

\begin{IEEEbiography}[{\includegraphics[width=1in,height=1.25in,clip,keepaspectratio]{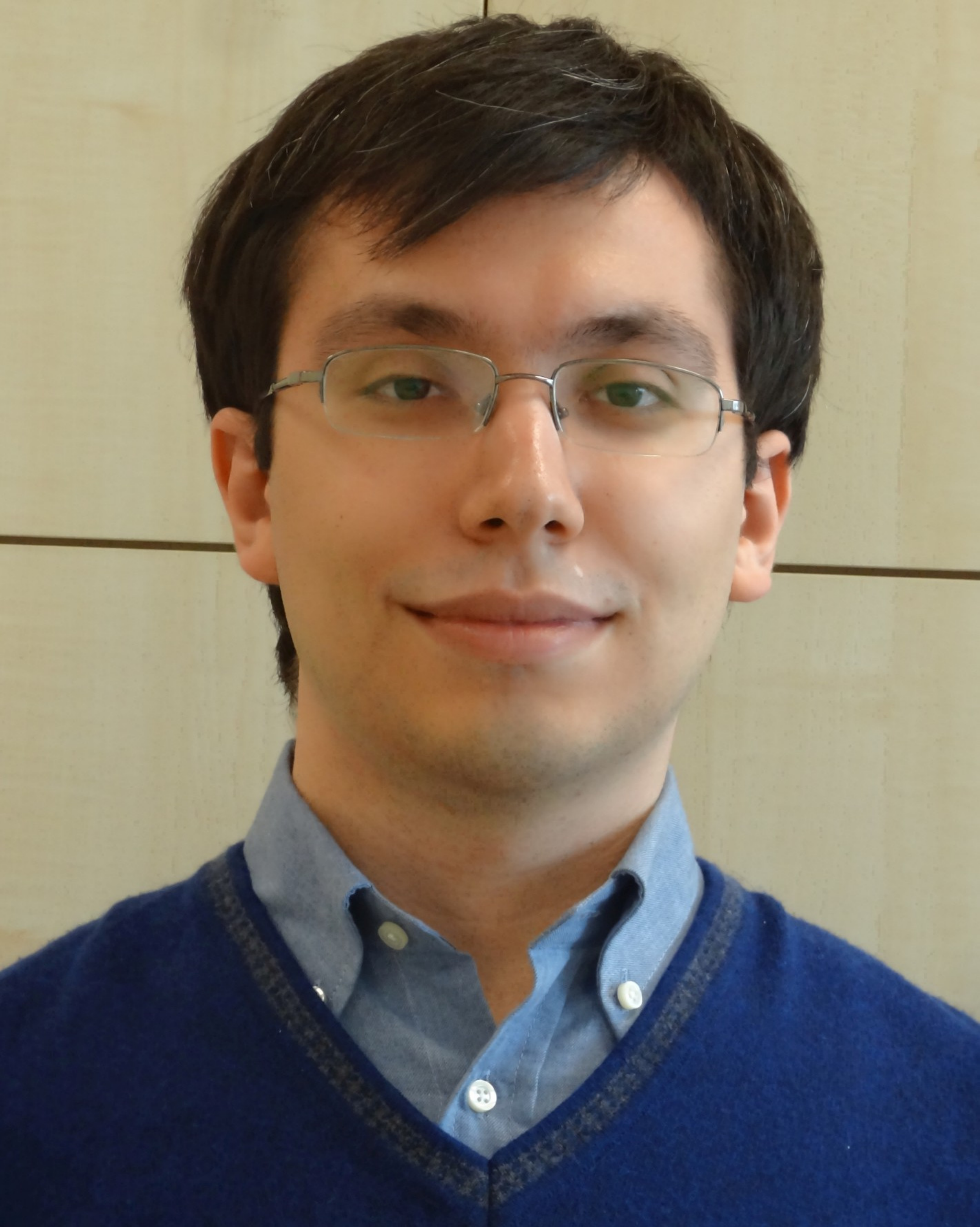}}]{Ahmet Dundar Sezer} received the B.S., M.S., and Ph.D. degrees in Electrical and Electronics Engineering from Bilkent University, Ankara, Turkey, in 2011, 2013, and 2018, respectively. He is currently a Post-Doctoral Researcher at the University of California, Santa Barbara, CA, USA. His current research interests include signal processing, wireless communications, and optimization.
\end{IEEEbiography}

\begin{IEEEbiography}[{\includegraphics[width=1in,height=1.25in,clip,keepaspectratio]{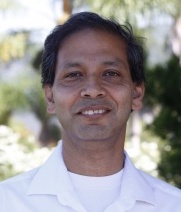}}]
{Upamanyu Madhow} is Professor of Electrical and Computer Engineering
at the University of California, Santa Barbara. 
His current research interests focus on next generation communication, sensing and inference infrastructures centered around millimeter wave systems, and on robust machine learning. 
He received his bachelor's degree in electrical engineering from the Indian Institute of Technology, Kanpur, in 1985, and his Ph. D. degree
in electrical engineering from the University of Illinois,
Urbana-Champaign in 1990. He has worked as a research scientist at
Bell Communications Research, Morristown, NJ, and as a faculty at the
University of Illinois, Urbana-Champaign.  Dr. Madhow is a recipient of the 1996 NSF CAREER award, and co-recipient of the 2012 
IEEE Marconi prize paper award in wireless communications. He has served as Associate Editor for the IEEE Transactions on
Communications, the IEEE Transactions on Information Theory, and the 
IEEE Transactions on Information Forensics and Security. He is 
the author of two textbooks published by Cambridge University Press, {\it Fundamentals of Digital Communication} (2008) and {\it Introduction to Communication Systems} (2014).
\end{IEEEbiography}



\vfill


\end{document}